\documentclass[journal]{IEEEtran}

\DeclareTextSymbol{\degre}{T1}{6}
\DeclareTextSymbol{\degre}{OT1}{23}
\usepackage{url}
\usepackage{amsmath,amsfonts,amssymb,amsthm,latexsym}
\usepackage[numbers,sort&compress]{natbib}
\usepackage{mathrsfs}
\usepackage{stmaryrd}
\usepackage{amscd}
\usepackage{mathtools}
\usepackage{bm}
\usepackage[T1]{fontenc}
\usepackage{graphicx}
\usepackage{epstopdf}
\usepackage[usenames, dvipsnames]{xcolor} 
\usepackage{tikz}
\usetikzlibrary{shapes}
\usepackage{dsfont}
\usepackage{graphicx}
\usepackage{bm}
\usepackage{tikz}
\usetikzlibrary{calc}
\usetikzlibrary{plotmarks}
\usetikzlibrary{shapes,arrows}
\usepackage[ruled]{algorithm2e}
\usepackage{yhmath}
\usepackage{graphics}
\usepackage{textcomp}
\usepackage{amssymb}
\usepackage{amscd}
\usepackage{epsfig}
\usepackage{psfrag}
\usepackage{rotating}
\usepackage{color}

\newcommand{\mc}{\mathcal}

\newtheorem{theorem}{Theorem}[section]
\newtheorem{definition}[theorem]{Definition}

\newtheorem{proposition}[theorem]{Proposition}
\newtheorem{corollary}[theorem]{Corollary}

\newtheorem{assumption}{Assumption}

\begin{document}
\title{Composite charging games in networks of electric vehicles}
\author{Olivier Beaude, Cheng Wan, and Samson Lasaulce
\date{\today}
\thanks{O. Beaude is PhD student at Renault SAS, L2S, and Supelec, C. Wan is with the Department of Economics and Nuffield College, University of Oxford, and S. Lasaulce is with L2S (CNRS - Supelec - Univ. Paris Sud XI), Paris, France (e-mail: olivier.beaude@lss.supelec.fr).}}

\markboth{NETGCOOP 2014}%
{Shell \MakeLowercase{\textit{et al.}}: Bare Demo of IEEEtran.cls for Journals}

\maketitle

\begin{abstract}
An important scenario for smart grids which encompass distributed
electrical networks is given by the simultaneous presence
of aggregators and individual consumers. In this work, an aggregator is seen as an entity (a coalition) which is able to manage jointly the energy demand
of a large group of consumers or users. More precisely, the demand 
consists in charging an electrical vehicle (EV) battery. The way the EVs user charge their batteries matters since it strongly impacts the network, especially the distribution network costs (e.g., in terms of Joule losses or transformer ageing). Since the charging policy is chosen by the users or the aggregators, the charging problem is naturally distributed. It turns out that one of the tools suited to
tackle this heterogenous scenario has been introduced only recently namely, through the notion of composite games. This paper exploits for the first time in the literature of smart grids the notion of composite game and
equilibrium. By assuming a rectangular charging profile for an
 EV, a composite equilibrium analysis is conducted, followed by a detailed analysis of a case study which assumes three possible charging periods or time-slots. Both the provided analytical and numerical results allow one to better understand the relationship between the size (which is a measure) of the coalition and the network sum-cost. In particular, a social dilemma, a situation where everybody prefers unilaterally defecting to cooperating, while the consequence is the worst for all, is exhibited.
\end{abstract}


\IEEEpeerreviewmaketitle

\section{Introduction}
\label{sec:introduction}

EV charging can lead to significant impacts on the existing and future energy networks \cite{Clement2009,Gong2011}. Considering different physical metrics, the smart grid literature pursued the goal to mitigate these impacts optimizing EV charging schedules using Demand Side Management \cite{Galus2008}, proposing charging algorithms \cite{Deilami2011,Shinwari12} or pricing policies \cite{Wu2012,Chao2010}.

The EV charging problem can be seen as a distributed problem. Indeed, it is reasonable to assume that EV's owners can decide when they plug their vehicles and how they charge their batteries. As a consequence, game theoretical tools have been proposed to tackle the charging problem (see e.g., \cite{Saad2012} for a recent survey). Very relevant contributions include \cite{Agarwal2011,Mohsenian-Rad2010,Wu2012,Ibars2010}. On the contrary, it may also be assumed that the charging profiles of the EVs are decided by a coordinator, often called aggregator \cite{Han2010,Wu2012}, who is much more informed about the real constraints of the electrical network. In this case, the decision is centralized and optimization tools are used to find the optimal policy. The present work deals with the situation in between, when a fraction of the EVs is supposed to independently decide their charging policies, while the rest is governed by an aggregator. This framework has been recently introduced in the game theoretic literature \cite{Wan11} with the notion of \emph{composite equilibrium} and, to the best of our knowledge, our work is the first to propose an application of this concept to the EV charging problem, or more generally to a smart grid issue.

The contributions of this paper include

\begin{itemize}
\item the formulation of the EV charging problem as a composite game,
\item the characterization and proof of existence of its equilibrium,
\item the description of its main properties in the particular case of three time-slots,
\item the numerical analysis of these properties in the case of distribution network costs.
\end{itemize}

The paper is structured as follows. Sec. \ref{sec:composEq} provides the model of composite games in the context of EV charging. Sec. \ref{sec:defComposEq} defines and characterizes the composite equilibrium. Sec. \ref{sec:partCase} treats a particular case, and conducts an equilibrium analysis. A thorough numerical analysis is then provided in Sec. \ref{sec:simu} and the paper is concluded by Sec. \ref{sec:ccl}.

\medskip

\emph{Notations.} Bold symbols stand for vectors.
For all $d \in \mathbb{N}^*$, $\Delta^d=\{\bm{z}=(z_i)_{i=1}^d\in \mathbb{R}^{d}|\bm{z}\geq \bm{0}, \sum^{d}_{i=1} z_i=1 \}$; for all $D>0$, $\Delta^d_{D}=\{\bm{z}=(z_i)_{i=1}^d\in \mathbb{R}^{d}|\bm{z}\geq \bm{0}, \sum^{d}_{i=1} z_i=D \}$.
\medskip

For all $m, n\in \mathbb{N}$ such that $m \leq n$, $\llbracket m,  n \rrbracket = \{m, m+1, \ldots, n-1, n\}$.

\medskip

For $\bm{x}=(x_1,x_2,...,x_d),\bm{x'}=(x'_1,x'_2,...,x'_d) \in \mathbb{R}^{d}$, $\left\langle \bm{x}, \bm{x'}\right\rangle=\sum^{d}_{i=1}x_i x'_i$ denotes the inner product of $\bm{x}$ and $\bm{x'}$.


\section{Composite EV charging game}
\label{sec:composEq}

\subsection{EV charging game}\label{EVChargGame}

There are $T$ time-slots, labeled by $t \in \mc{T}=\lbrace 1, 2, \ldots, T \rbrace$. We consider a set of EVs which have to choose $C$ consecutive time-slots, $C \leq T$, to charge. Each EV is considered as a player, i.e. aims to minimize his charging cost taking into account the impact of the other EVs' charging decisions using the framework of Game Theory \cite{Beaude2012}. The assumption of choosing consecutive time-slots is mathematically restrictive\footnote{Indeed, the charging vector is only defined by the time to start charging while a more general case would be to consider charging vectors in the $T-1$-dimensional simplex.} and thus provides a less general mathematical structure than with freely varying charging vectors \cite{Orda1993}.
It is nonetheless very important to mitigate impact of the charging policy on the battery lifetime because highly time varying charging currents can increase battery temperature much more and lead to a shorter lifetime \cite{Gan2012}.

\medskip

\subsection{Composite game}\label{CompGame}

In this paper, we consider EVs as nonatomic players who have weight zero and who are called \emph{individuals}. Nonatomic means that an EV alone cannot have an impact on the cost perceived by other players. This holds in particular when there is a large number of EVs. A group of individuals of positive total weight form a \emph{coalition} if the behaviour of its members is coordinated by an \textit{aggregator} \cite{Wu2012}. A charging game with both nonatomic players and coalitions is called a \emph{composite} game \cite{Wan11} and its equilibria are called \emph{composite} ones accordingly. This corresponds to the practical situation where some of EVs decide their charging policies independently, while the others are coordinated by one or several aggregators. Without loss of generality, the total weight of all the EVs is assumed to be $1$.

\subsection{Charging flows definition}\label{ChargingFlows}

We suppose that there are $K$ coalitions. For any coalition $k$ of size $M^k>0$ ($1\leq k \leq K$), let $x^k_t$ denote the weight of the EVs sent by it to start charging at time $t\in \llbracket 1, T-C+1 \rrbracket$. Thus the vector $\bm{x}^k=(x^k_t)^{T-C+1}_{t=1}\in F^k=\Delta^{T-C+1}_{M^k}$ characterizes the strategy, or \emph{charging flow}, of coalition $k$.

For $t\in \llbracket 1, T \rrbracket$, let $y^{k}_{t}$ denote the quantity of EVs from coalition $k$ charging at time $t$, and define $\bm{y}^k=(y^{k}_{t})^T_{t=1}$. It is called the \emph{charging load} induced by its strategy $\bm{x}^k$. Indeed,
\begin{equation}
\label{eq:defLoad}
y^{k}_{t} = \displaystyle \sum_{s=\max(t-C+1,1)}^{t}x^{k}_{s}, \quad \forall\, 1 \leq t \leq T \textrm{.}
\end{equation}

Let $\hat{F}^k$ be the set of charging loads of coalition $k$ induced by its strategies. It is not difficult to verify that $\hat{F}^k$ is a convex and compact subset of $\mathbb{R}^{T}$.

\medskip

For the individuals of total weight $M^0\geq 0$ ($M^0=1-\sum_{k=1}^K M^k$), let $x^0_t$ be the weight of the individuals starting charging at time $t\in \llbracket 1, T-C+1 \rrbracket$. Define $\bm{x}^0=(x^0_t)^{T-C+1}_{t=1} \in F^0=\Delta^{T-C+1}_{M^0}$, which can be viewed as the strategy profile of the individuals.

Let $\bm{y}^{0}$ be the \emph{charging load} induced by $\bm{x}^0$ and $\hat{F}^{0}$ the set of charging loads induced by the strategies of the individuals.

\medskip

Finally, let $\bm{x}=(\bm{x}^i)^{K}_{i=0}\in F=F^0\times F^1 \times \cdots \times F^K$ denote the strategy profile of all the players and $\bm{y}=(\bm{y}^i)^{K}_{i=0}$ be the \emph{system charging load}. Denote $\hat{F}=\hat{F}^0 \times \hat{F}^1 \times \cdots \times \hat{F}^K$ the set of feasible system charging loads.

\medskip

The total weight of EVs charging at time $t$ is denoted by

\begin{equation}
\label{eq:defChargingWeight}
z_t:=\sum_{k=0}^{K}y^{k}_{t} \text{ .}
\end{equation}

Define the \emph{aggregate charging load} as $\boldsymbol{z}=(z_t)^{T}_{t=1}$.

\subsection{Cost definition}\label{CostDef}

Let us now introduce the (per unit) charging cost function at each time-slot, $f$. If there are EVs of total weight $z_t$ charging at time $t$, then the charging cost for each of them is $f(L_{t}+P z_t)$ during that time-slot, where 

\begin{itemize}
\item $\boldsymbol{L}=(L_t)^{T}_{t=1}$ is the non-EV load, assumed to be known;
\item $P$ the total EV charging power in the district, i.e. the number of EVs in the district multiplied by the charging rate of an EV assumed constant between EVs here for simplicity\footnote{Typically $3$kW in the residential case.};
\item and $f$ is a real-valued function defined on $[0,W]$ for $W$ sufficiently large, i.e. bigger than the potential maximal load in the district.
\end{itemize}

This expresses the fact that EV charging cost depends on the impact measured in the elctrical network, which is one of the main ideas of smart grids. Notice that $f$ is common for all the EVs.

\begin{assumption}\label{assp_1}
The charging cost function $f$ is of class $\mathcal{C}^1$, convex, strictly increasing and nonnegative on $[0,W]$.
\end{assumption}
Assumption \ref{assp_1} always holds in this paper.

For an individual EV, the charging strategy consists in charging from time-slot $t$ to time-slot $t+C-1$ ($t\in \llbracket 1, T-C+1 \rrbracket$), and its cost function is
\begin{equation}
u_{t}(\bm{x})=\displaystyle \sum_{s=t}^{t+C-1} f(L_s+Pz_s), \forall \, \bm{x} \in F,\; \forall \, 1 \leq t \leq T-C+1\textrm{ ,}
\end{equation}
where the dependency between $\bm{x}$ and $\bm{z}$ is implicit in the notations and comes from (\ref{eq:defLoad}) and (\ref{eq:defChargingWeight}). The average cost to coalition $k$ can be written as a function of $\bm{x}$
\begin{equation}
\Pi^{k}(\bm{x})=\frac{1}{M^k}\displaystyle \sum_{t=1}^{T-C+1} x^{k}_{t}\,u_{t}(\bm{x}) \textrm{ ,}
\end{equation}

or as a function of $\bm{y}$
\begin{equation}
\hat{\Pi}^{k}(\bm{y}) = \frac{1}{M^k}\displaystyle \sum_{t=1}^{T} y^{k}_{t}\,f(L_t+Pz_t) \textrm{ .}
\end{equation}

The average cost to the individuals can be similarly defined as a function of flow $\bm{x}$, denoted by $\Pi^0$, or as a function of load $\bm{y}$, denoted by $\hat{\Pi}^0$.

Finally, the social cost is function of $\boldsymbol{y}$

\begin{equation}
\Pi(\boldsymbol{y}) = \displaystyle \sum_{t=1}^{T} z_{t}\,f(L_t+Pz_t)  \textrm{ ,}
\end{equation}
where, again, the dependancy between $\bm{z}$ and $\bm{y}$ is implicit in the notations.

Let this (composite EV) charging game be denoted by $\mathcal{G}(T,C,f,(M^i)_{i=0}^K)$.

\medskip

\subsection{Application to distribution network costs}\label{AppliDN}

In the context of residential distribution networks (the system which delivers power from the generation points to the end users, see \cite{Clement2009} for an illustration), two particular classes of physical cost functions $f$ will be considered:

\begin{itemize}
\item Joule losses, which is a quadratic one: $f(L_{t}+P z_t)=(L_{t}+P z_t)^{2}$ \cite{Deilami2011};
\item equipment ageing (transformers for example), which can be approximated by an exponential function (when the transformer power is close to its nominal level) \cite{Gong2011}: $f(L_{t}+P z_t)=\exp[\beta(L_{t}+P z_t)]$, with $\beta>0$.
\end{itemize}

Note that the standard "mathematical" linear case will also be studied, which corresponds to a standard approximation of practical applications. Observe that Assumption \ref{assp_1} holds in these three cases and also when considering a weighted sum of these objectives.

\section{Definition and characterization of composite equilibrium}
\label{sec:defComposEq}

We are now interested in defining and characterizing a configuration of equilibrium at which neither the individuals nor the coalitions have incentive to deviate from their current strategy.

\subsection{Composite equilibrium conditions} \label{CompositeEqCond}

At equilibrium, if strategy $t$ is used by individuals, i.e. $x^{0}_{t}>0$, then
\begin{equation}\label{eqcnd_str}
 u_{t}(\bm{x}) \leq u_{s}(\bm{x}), \qquad \forall s \in \left\{1,2,...,T-C+1\right\}\, \textrm{ ,}
\end{equation}
according to the standard Wardrop equilibrium conditions \cite{War52}. Notice that in this case the cost of all the individuals is common and equal to the average cost $\Pi^0$.

At equilibrium, a coalition $k$ minimizes the average cost of its members\footnote{This is obviously equivalent to minimizing the total cost of its members.}, given the strategies of the other coalitions and individuals:
\begin{equation}\label{eqcnd_coal}
\Pi^{k}(\bm{x}^{k},\bm{x}^{-k})=\min_{\bm{x'}^{k}\in F^{k}}\Pi^{k}(\bm{x'}^{k},\bm{x}^{-k})
\end{equation}
where $\bm{x}^{-k}=(\boldsymbol{x}^i)_{0 \leq i \leq K, \, i \neq k}$ denotes the strategies of other players than coalition $k$.

\begin{definition}
In a composite charging game $\mathcal{G}(T,C,f,(M^i)_{i=0}^K)$, a configuration $\bm{x}\in F$ is a composite equilibrium (CE) if conditions \eqref{eqcnd_str} and \eqref{eqcnd_coal} are satisfied.
\end{definition}
\medskip

Composite equilibria can also be characterized via variational inequalities. Denote the gradient of coalition $k$'s cost w.r.t. its strategy (or flow) $\boldsymbol{x}^k$ by
\begin{equation}
\bm{U}^{k}(\bm{x}^k, \bm{x}^{-k}) := \nabla_{\bm{x}^{k}}\Pi^{k}(\bm{x}^k, \bm{x}^{-k}) \textrm{.}
\end{equation}
Define also $\bm{U}^{0}(\bm{x})=\Pi^0(\bm{x})$ and let $\bm{U}=(\bm{U}^i)^K_{i=0}$.

At the minimum, the \textbf{first order (necessary) condition} of the minimization problem (\ref{eqcnd_coal}) is
\begin{equation}\label{vip_coal}
\forall \tilde{\bm{x}}^k\in F^k, \, \quad \left\langle \bm{U}^{k}(\bm{x}^k, \bm{x}^{-k}),\tilde{\bm{x}}^k-\bm{x}^k \right\rangle \geq 0 \textrm{.}
\end{equation}

\begin{proposition}\label{prop:CEineqvar}
If for all coalition $k$, $\Pi^k(\bm{x}^k,
\bm{x}^{-k})$ is convex with respect to $\bm{x}^k$ on $F^k$ for all $\bm{x}^{-k}\in F^{-k}$, then $\bm{x}^{*}\in F$ is a composite equilibrium if and only if
\begin{align*}
&\left\langle \bm{U}^{0}(\bm{x}^{* 0}, \bm{x}^{* -0}),\bm{x}^0-\bm{x}^{* 0} \right\rangle \geq 0, \, \,  \forall \bm{x}^0\in F^0 \textrm{,}\\
&\left\langle \bm{U}^{k}(\bm{x}^{* k}, \bm{x}^{* -k}),\bm{x}^k-\bm{x}^{* k} \right\rangle \geq 0, \, \, \forall \bm{x}^k\in F^k, \forall\, k=1,\ldots, K \textrm{,}
\end{align*}
which is equivalent to
\begin{equation}\label{eq:varineq}
\left\langle \bm{U}(\bm{x}^{*}),\bm{x}-\bm{x}^{*} \right\rangle \geq 0, \, \quad  \forall\, \bm{x}\in F \textrm{.}
\end{equation}
\end{proposition}

\begin{proof}
The proof is similar to the one of Prop. 1 in \cite{Wan11}.
\end{proof}

\bigskip

It can be shown that the necessary and sufficient condition (\ref{eq:varineq}) remains the same if the cost function $f$ undergoes an affine transformation $f \longmapsto af+b$ with $a>0, \, b \in \mathbb{R}$. In turn, $\bm{x}^{*}$ is a composite equilibrium of $\mathcal{G}(T,C,f,(M^k)_{k=0}^K)$ if and only if it is a composite equilibrium of $\mathcal{G}(T,C,f+b,(M^i)_{i=0}^K)$ (or $\mathcal{G}(T,C,af,(M^i)_{i=0}^K)$ for $a>0$). This is why the simulations at the end of this paper are restricted to the case $f(L_t+Pz_t)=(L_t+Pz_t)^2$ in the quadratic case and $f(L_t+Pz_t)=L_t+Pz_t$ in the linear case.

With formulation (\ref{eq:varineq}), the existence of an equilibrium in the charging game $\mathcal{G}(T,C,f,(M^k)_{k=0}^K)$ is easily obtained.

\begin{theorem}
\label{thm:existCE}
If for all coalition $k$, $\Pi^k(\bm{x}^k,
\bm{x}^{-k})$ is convex with respect to $\bm{x}^k$ on $F^k$ for all $\bm{x}^{-k}\in F^{-k}$, then the charging game $\mathcal{G}(T,C,f,(M^k)_{k=0}^K)$ has a composite equilibrium.
\end{theorem}

\begin{proof}
Because of the continuity of $f$, $\bm{U}$ is continuous on the compact and convex set $\bm{F}$. The variational inequality (\ref{eq:varineq}) thus has a solution \cite{Kin86}. Given Prop. \ref{prop:CEineqvar}, this (eventually these) solution(s) is (are) a composite equilibrium(s).
\end{proof}

\bigskip

We also have an important property concerning the comparison between the cost of individuals, the average cost in a coalition and the social cost. Indeed, by the definition of CE, individuals choose the least expensive alternative(s) at the CE, therefore one obtains the following general result immediately.
\begin{proposition}
\label{prop:ordCost}
At a composite equilibrium,
\begin{equation}
\Pi^{0}(\!M\!)\leq \Pi(\!M\!) \leq \Pi^k(\!M\!) ,\quad \forall M \in (0,1], \ \forall k, \ 1 \leq k \leq K.
\end{equation}
\end{proposition}

\section{Particular case of three time-slots}
\label{sec:partCase}

In this section, the properties of composite equilibrium are studied for a charging game $\mathcal{G}(T,C,f,(M^i)_{i=0}^K)$ with $T=3$ time-slots, a charging duration $C=2$, a unique coalition of size $M\in (0,1]$ and a group of individuals of total weight $1-M$. This corresponds to a situation where there are a peak, an off-peak and a standard time-slot. This also suits the case of charging places where EVs do not stay a long time, e.g. a parking. Indeed, because the charging rate will not vary with a small time step, this leads to situations where the number of time periods considered is very small. Finally, this is a first step to determine some intuitive theoretical results that could be then observed by simulations of cases with a bigger number of time-slots.

\subsection{Notations of this particular case}
\label{subsec:notPartCase}

An EV has two alternatives:
\begin{itemize}
\item alternative $1$: charging at $t=1$ and $t=2$;
\item alternative $2$: charging at $t=2$ and $t=3$.
\end{itemize}

For the individuals, these are just their strategies. In the coalition, the aggregator assigns an alternative to each EV.

\medskip

Time indexes are omitted to simplify the notations. Let $x^{1}$ (resp. $M-x^{1}$) denote the weight of EV to which the aggregator assigns alternative $1$ (resp. alternative $2$), and $x^{0}$ (resp. $1-M-x^{0}$) the weight of the individuals choosing alternative $1$ (resp. alternative $2$). Therefore, $x^{1}\in [0,M]$, $x^{0}\in [0,1-M]$, and they respectively characterize the strategy of the coalition and the choices of the individuals. To simplify the notations, we set here $P=1$ without changing the essential nature of the results provided here\footnote{This can be done without loss of generality because this is equivalent to scaling $\bm{L}$.}. Since all the EVs are charging at time-slot $t=2$, the per-unit cost for this time-slot, $f(1+L_2)$, is common to all the EVs. Thus, we simply need to study the following costs

\begin{equation}
\begin{cases}
\tilde{\Pi}^0=\Pi^0-f(1+L_2) \\
\tilde{\Pi}^1=\Pi^1-f(1+L_2) \\
\tilde{\Pi}=\Pi-f(1+L_2) \\
\end{cases}
\textrm{.}
\end{equation}

Considering the choice made by the coalition, $\tilde{\Pi}^1$ is a function of $(x^{0},x^{1})$, defined on $[0,1]^2$: 
\begin{scriptsize}
\begin{equation*}
\tilde{\Pi}^{1}(x^{0},x^{1})=\frac{x^{1} f(L_{1}+x^{1}+x^{0}) +(M-x^{1})f(L_{3}+1-x^{1}-x^{0})}{M}.
\end{equation*}
\end{scriptsize}

For example, in the Joule losses case, $f(L)=L^2$,

\begin{scriptsize}
\begin{equation*}
\tilde{\Pi}^{1}(x^{0},x^{1})= \frac{x^{1} (L_{1}+x^{1}+x^{0})^2 +(M-x^{1})(L_{3}+1-x^{1}-x^{0})^2}{M}.
\end{equation*}
\end{scriptsize}

\subsection{Configuration of composite equilibrium}
\label{subsec:configCE}

Without loss of generality, suppose that $L_1 \geq L_3$. This means that the cost of the first time-slot (resp. alternative 1) is higher than that of the third time-slot (resp. alternative 2) without EV charging.

The explicit computation of the CE $x^*$ using the conditions (\ref{eqcnd_str}) and (\ref{eqcnd_coal}) is omitted. We summarize the results here. Observe that the individuals' common cost, $\Pi^0$, the coalition's average cost, $\Pi^1$, and the social cost $\Pi$ are the key variables for analyzing the efficiency of the CE.

\medskip

\paragraph{Case 1: $L_{1} \geq L_{3} + 1$}

\begin{itemize}
\item For $M \in (0, \frac{f(L_{1})-f(1+L_{3})}{f^{\prime}(1+L_{3})}]$ (or $M \in (0,1]$ if $\frac{f(L_{1})-f(1+L_{3})}{f^{\prime}(1+L_{3})}\geq 1$)
\begin{equation}\label{cfg_2}
x^{* 1} = 0,\quad x^{* 0} = 0
\end{equation}
\begin{equation}\label{cost_2}
 \tilde{\Pi}=\tilde{\Pi}^{0}=\tilde{\Pi}^1= f(1+L_{3})
\end{equation}
\item For $M \in (\frac{f(L_{1})-f(1+L_{3})}{f^{\prime}(1+L_{3})},1]$ (possible only if $\frac{f(L_{1})-f(1+L_{3})}{f^{\prime}(1+L_{3})}<1$)
\begin{equation}\label{cfg_1}
\begin{cases}
0<x^{* 1}<M:\\
f(L_{1}+x^{* 1})+x^{* 1}f^{\prime}(L_{1}+x^{* 1})\\
\hspace{0.3cm}=f(1+L_{3}-x^{* 1})+(M-x^{* 1})f^{\prime}(1+L_{3}-x^{* 1}) \\
x^{* 0} = 0.
\end{cases}
\end{equation}
\begin{equation}\label{cost_1}
\begin{cases}
\tilde{\Pi}  = x^{* 1} f(L_{1}+x^{* 1}) + (1-x^{* 1})f(1+L_{3}-x^{* 1})\\
 \tilde{\Pi}^{0} =f(1+L_{3}-x^{* 1})\\
 \tilde{\Pi}^1  = \frac{1}{M}\big[ x^{* 1}f(L_{1}+x^{* 1}) + (M-x^{* 1})f(1+L_{3}-x^{* 1})\big]
 \end{cases}
\end{equation}
\end{itemize}

\paragraph{Case 2: $L_{3} \leq L_{1} < L_{3}+1$}
\begin{itemize}
\item For $M < 1+L_{3}-L_{1} (\leq 1)$
\begin{equation}\label{cfg_3}
x^{* 1}=\frac{M}{2}, \quad x^{* 0}=\frac{1+L_{3}-L_{1}-M}{2}
\end{equation}
\begin{equation}\label{cost_3}
\tilde{\Pi}=\tilde{\Pi}^{0}=\tilde{\Pi}^1= f\big(\frac{1+L_{1}+L_{3}}{2}\big)
\end{equation}

\item For $M \geq 1+L_{3}-L_{1}(>0)$
\begin{equation}\label{cfg_4}
\begin{cases}
\frac{1+L_{3}-L_{1}}{2}<x^{* 1}<\frac{M}{2}:\\
f(L_{1}+x^{* 1})+x^{* 1}f^{\prime}(L_{1}+x^{* 1})\\
\hspace{0.3cm}=f(1+L_{3}-x^{* 1})+(M-x^{* 1})f^{\prime}(1+L_{3}-x^{* 1}) \\
x^{* 0} = 0
\end{cases}
\end{equation}
\begin{equation}\label{cost_4}
\begin{cases}
 \tilde{\Pi}  =x^{* 1} f(L_{1}+x^{* 1}) + (1-x^{* 1})f(1+L_{3}-x^{* 1})\\
 \tilde{\Pi}^{0} =f(1+L_{3}-x^{* 1})\\
 \tilde{\Pi}^1 = \frac{1}{M}\big[ x^{* 1} f(L_{1}+x^{* 1}) + (M-x^{* 1})f(1+L_{3}-x^{* 1})\big]
 \end{cases}
\end{equation}
\end{itemize}

\subsection{Properties of composite equilibrium}
\label{subsec:propCE}

Using the configuration of the CE obtained previously, its main properties, namely existence, uniqueness, and variation with the size of the coalition, are now investigated.

\begin{proposition}
For all $M\in (0,1]$, there exists a unique composite equilibrium $(x^{*0},x^{*1})$.
\end{proposition}
\begin{proof}
The existence results from Thm. \ref{thm:existCE} given that Assumption \ref{assp_1} holds. The uniqueness directly follows from the summary of the results given just before.
\end{proof}

\medskip

Observe also that the individuals' weight on the first alternative at CE is independent of the charging cost function verifying Assumption \ref{assp_1}.

Since the equilibrium $(x^{*0},x^{*1})$ is unique for each coalition size $M\in (0,1]$, one can now consider the following quantities as functions of $M$ and omit the superscript $*$: the quantity of EV charging from $t=1$ in the coalition, $x^{1}$, the quantity of individuals taking strategy 1, $x^{0}$, the individuals' common cost, $\Pi^0$, the coalition's average cost $\Pi^1$, and the social cost, $\Pi$.

Additional properties of the CE configuration $(x^0(M),x^1(M))$ can be deduced easily from the results in Section~\ref{subsec:configCE}. To this end, an additional hypothesis on the charging cost function will be needed.

\begin{assumption}\label{assp_2}
Charging cost function $f$ is of class $\mathcal{C}^2$ on $[0,W]$.
\end{assumption}

\begin{proposition}\label{prop:x}
Under Assumption \ref{assp_2}, function $x^{1}$ is continuous and increasing in $M$ on $(0,1]$. More precisely,
\begin{enumerate}
 \item if $L_{1} \geq L_{3} + 1$ and $\frac{f(L_{1})-f(1+L_{3})}{f^{\prime}(1+L_{3})}\geq 1$, then $x^{1}$ is constant on $(0,1]$;
 \item if $L_{1} \geq L_{3} + 1$ and $\frac{f(L_{1})-f(1+L_{3})}{f^{\prime}(1+L_{3})}< 1$, then $x^{1}$ is constant on $(0, \frac{f(L_{1})-f(1+L_{3})}{f^{\prime}(1+L_{3})}]$ and it is strictly increasing in $M$ on $(\frac{f(L_{1})-f(1+L_{3})}{f^{\prime}(1+L_{3})},1]$;
 \item if $L_{3} \leq L_{1} < L_{3}+1$, then $x^{1}$ is strictly increasing in $M$ on $(0,1]$.
\end{enumerate}
\end{proposition}
\begin{proof}
Both the continuity of $x^{1}$ on $(0,1]$ and the monotonicity are obtained by using the implicit function theorem on the equations characterizing $x^{1}$.
\end{proof}

\medskip

The proofs of the following results of this paper, which are also based on the application of the implicit function theorem, are omitted. Considering the weight of individuals using strategy $1$, a similar result, but with the opposite monotonicity, is obtained.

\medskip

\begin{proposition}\label{prop:y}
Function $x^{0}$ is continuous and decreasing in $M$ on $(0,1]$. More precisely,
\begin{enumerate}
 \item if $L_{1} \geq L_{3} + 1$, then $x^{0}$ is constant ($x^{0}=0$) on $(0,1]$;
 \item if $L_{3} \leq L_{1} < L_{3}+1$, then $x^{0}$ is strictly (linearly) decreasing in $M$ on $(0,1+L_{3}-L_{1})$ and it is constant ($x^{0}=0$) on $[1+L_{3}-L_{1},1]$.
\end{enumerate}
\end{proposition}

An intuitive interpretation of Prop.\ref{prop:x} and Prop.\ref{prop:y} is as follows. The bigger the coalition, the more EV it puts on the more expensive charging alternative, and the less individuals who choose the more expensive alternative. This highlights that bigger coalitions integrate more externalities. When the coalition is of size one, this leads to the social optimum.

\medskip

The following proposition now characterizes the convexity of the EV put on strategy $t=1$ at equilibrium, $x^1$.

\begin{proposition}\label{prop:x_order2}
If $f$ is linear $x\mapsto x$, quadratic $x\mapsto x^2$, or exponential $x\mapsto e^{\beta x}$ $(\beta>0)$, then function $x^{1}$ is concave in $M$ on $(0,1]$.
\end{proposition}

This proposition expresses a phenomenon of saturation. Although the coalition puts more EVs on the more expensive alternative when its size increases, the additional weight on this alternative decreases with respect to its size. Note that some properties for the function $x^1$ such as its linearity or strict concavity are available for some choices of the parameters $L_1$ and $L_3$ as for Prop. \ref{prop:x}.

\subsection{Cost at the composite equilibrium}
\label{subsec:CostGenCase}

First of all, according to \eqref{cfg_2}, \eqref{cfg_1}, \eqref{cfg_3} and \eqref{cfg_4}, the continuity of $x^0$ and $x^1$ leads to that of $\Pi^0$, $\Pi^1$ and $\Pi$ on their domains of definition. In particular,  \eqref{cfg_1} and \eqref{cfg_3} show that $\Pi^0$ can be extended to $M=1$ in a continuous way. Thus, one has the following corollary of Prop. \ref{prop:x}.
\begin{corollary}
Under Assumption \ref{assp_2}, the individual's cost $\Pi^0$, the average cost of the coalition $\Pi^1$, and the social cost $\Pi$ at the CE are continuous in $M$ on $(0,1]$.
\end{corollary}

After this technical property, two issues arise to better understand the influence of the coalition on the costs:
\begin{enumerate}
\item comparison between the cost of individuals, the average cost in the coalition and the social cost;
\item analysis of the monotonicity of these costs with respect to the size of the coalition $M$.
\end{enumerate}

Concerning the first issue, the answer is given in the general case by Prop. \ref{prop:ordCost}: the cost of individuals is smaller than the social cost which is itself smaller than the average cost in the coalition. Let us now focus on the impact of the size of the coalition on $\Pi^0$, $\Pi$, and $\Pi^1$. This is of primary interest for a social planner who tries to analyze the cost of the different entities under its supervision and to decide if it is worthwhile to encourage the formation of coalitions.

\begin{proposition}
Under Assumption \ref{assp_2}, the individual's cost $\Pi^0$, the average cost of the coalition $\Pi^1$, and the social cost $\Pi$ at the CE are decreasing in the size of the coalition $M$.
\end{proposition}

The previous proposition shows that the bigger the coalition is, the better it is for everyone. However, considering Prop. \ref{prop:ordCost}, this leads to a "social dilemma", a situation in which collective interests are at odds with private interests. Indeed, the social optimum is attained when the coalition is global, i.e. $M=1$; meanwhile, according to Prop. \ref{prop:ordCost}, each EV prefers to act individually, leading to $M=0$. This phenomenon is similar to the one studied in \cite{Wan11}. Designing incentives to encourage EV to join the coalition could constitute a relevant extension of this work.

\section{Simulation results}
\label{sec:simu}

\subsection{Quantifying the results of the particular case $T=3$, $C=2$}
\label{subsec:simuT3C2}

We first investigate the particular case $T=3$, $C=2$. The non-EV load is supposed of the form $\bm{L}=(L_1,1,1)$ with $L_1 \geq 1$ representing the load at peak time of other electrical consumptions, $P=1$, and $\beta=1$ when the exponential charging cost function $f$ is considered. In turn, this provides a thorough application of the theoretical framework presented in the previous part. The CE configuration is thus directly given, by (\ref{cfg_2}) for example, or calculated by solving implicit equations, (\ref{cfg_1}) for example. The variation of the configuration of the CE and that of the equilibrium costs associated will be analyzed and quantified according to the charging cost function $f$ and the size of the coalition $M$.

Fig. \ref{fig:x1} and \ref{fig:x1bis} present the weight $x^1$ put by the coalition on alternative $1$ for  the three standard charging cost functions: linear, quadratic, and exponential. As theoretically claimed, when $L_1 \geq L_3 +1$ (cf. Fig. \ref{fig:x1}), $x^1$ is zero for small values of $M$, and it becomes positive from different thresholds of $M$ for different metrics $f$. When $L_1<L_3+1$ (cf. Fig. \ref{fig:x1bis}), $x^1$ is the same for all the metrics $f$ up to a common threshold $\bar{M}=1+L_3-L_1$, after which $x^1$ is different for different metrics. Also, one observes that $x^1$ is greater when $L_1$ is relatively lower or, equivalently, when time-slot $1$ is relatively less expensive. Take the linear cost metric as example. Fixing $L_3=1$, if $L_1=2.3$ (Fig. \ref{fig:x1}), $x^1$ remains $0$ till $M=0.3$ then it increases linearly to $0.18$ when $M=1$; while if $L_1=1.5$ (Fig. \ref{fig:x1bis}), $x^1$ increases in a piecewise linear manner from $0$ to $0.37$ while $M$ varies from $0$ to $1$.

\begin{figure}[!htbp]
\includegraphics[scale=0.6]{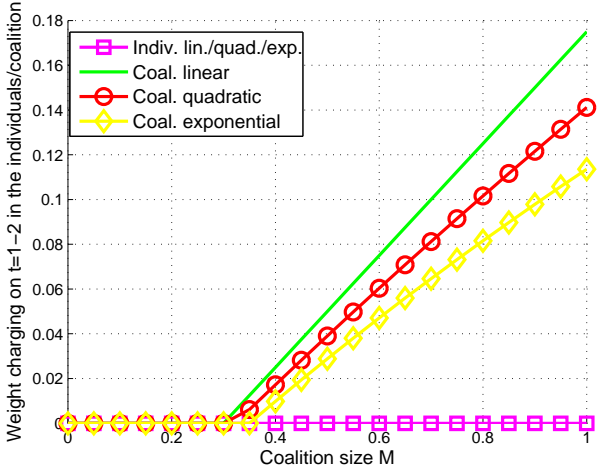}
\caption{\label{fig:x1} Configuration of CE according to the size of the coalition $M$ for $L_1=2.3 \geq L_3+1$: \textit{the individuals do not use the most expensive alternative $t=1-2$}.}
\end{figure}

\begin{figure}[!htbp]
\includegraphics[scale=0.6]{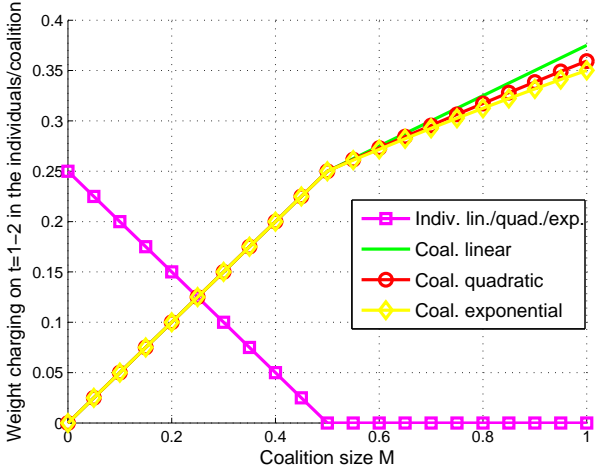}
\caption{\label{fig:x1bis} Configuration of CE according to the size of the coalition $M$ for $L_3 \leq L_1=1.5<L_3+1$.}
\end{figure}

\medskip

Next, the cost of CE is analyzed in the case of an exponential or quadratic charging cost function. Fig. \ref{fig:cost} presents the different normalized costs, i.e. the costs divided by the social cost at $M=0$. It shows the ranking between individual, coalitional and social costs and their monotonicity. It also quantifies the social benefit realized with respect to the size of the coalition. In the quadratic case, for example, a small gain of approximately $3\%$ is made with a coalition of size $M=1$ in comparison with the case with only individuals ($M=0$). However, in the situation of a global coalition ($M=1$), this figure also shows that any EV deviating to schedule alone its charging policy will do a significant benefit of $\frac{0.97-0.88}{0.97} \approx 9\%$. This highlights that the configuration with a coalition of size one is very efficient but also very unstable in the sense that each individual EV has a great interest to quit the coalition.

\begin{figure}[!htbp]
\includegraphics[scale=0.6]{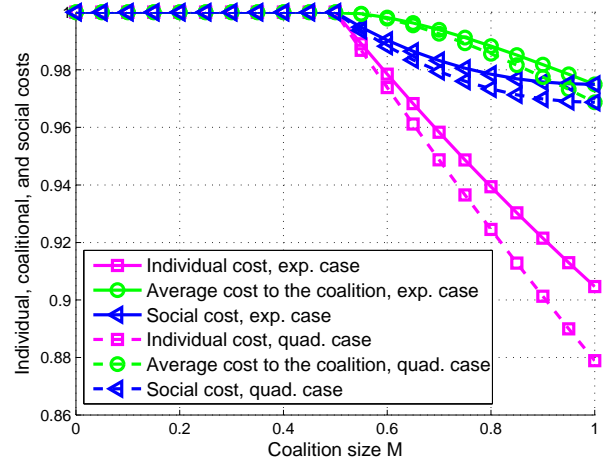}
\caption{\label{fig:cost} Individual, coalitional and social cost at CE for a quadratic or exponential charging cost function and $L_1=1.5<L_3+1$.}
\end{figure}

\subsection{A first step towards larger dimensions}

As one would expect, determining the CE of a game is rather complicated. Even for some given charging cost function $f$ and for small instances, it may be impossible to give the explicit form of the equilibrium configuration analytically. Consequently, it is of interest to see whether there are simple and distributed learning schemes that allow players to arrive at a reasonably stable solution. One of these schemes is based on an exponential learning behavior where players play the game repeatedly and learn the best strategies by keeping record of their strategies' performance (see \cite{Mertikopoulos09}). At each step denoted by index $n$, the individuals update their cumulative cost for strategy $t$, $V_t^{0,(n)}$, as

\begin{equation}
\label{eq:repDyn}
V_t^{0,(n)}=V_t^{0,(n-1)}+u_t(\bm{x}^{(n-1)}) \text{ ,}
\end{equation}
where $\bm{x}^{(n-1)}=(\bm{x}^{0,(n-1)},\bm{x}^{1,(n-1)})$ is the strategy profile of all the players at the $(n-1)$th iteration of the dynamics. These cumulative costs reinforce the perceived success of each strategy as measured by the average payoff it yields. Hence, the players will lean towards the strategy with the smallest cumulative cost. The precise way in which this is done is by playing according to the exponential law:

\begin{equation}
\label{eq:dynRep2}
x_t^{0,(n)}=\frac{e^{-V_t^{0,(n)}}}{\sum_{s=1}^{T-C+1} e^{-V_s^{0,(n)}}} \text{ .}
\end{equation}

Similarly, the coalition updates its cumulative cost $V_t^{1,(n)}$ for strategy $t$ replacing $u_t(\bm{x}^{(n-1)})$ by $\nabla_{\bm{x}^{1}}\Pi^{1}(\bm{x}^{(n-1)})$ in (\ref{eq:repDyn}) and its weights according to (\ref{eq:dynRep2}) with $V_t^{1,(n)}$ instead of $V_t^{0,(n)}$.



When players update their cumulative costs in continuous time, we obtain the standard replicator dynamics \cite{Mertikopoulos09}. Interestingly, this dynamics has been shown to converge\footnote{Furthermore, if the convergence point is an interior point, it is a composite equilibrium.} for composite games in the case of linear cost functions \cite{Cominetti09}.

First, this dynamics has been tested on the simple cases analyzed in Sec. \ref{subsec:simuT3C2} and we find the same results as the ones obtained with the analytical formula, not only for the linear case for which it is theoretically proven but also for the quadratic and exponential cases. This is of primary interest for practical applications and also leads to the open problem of the convergence in the quadratic and exponential cases.

Then, we propose a first realistic application studying the EV charging during the night time in a district with $T=7$ considering a two hours time step; $t=1$ corresponds to $5pm-7pm$, $t=2$ to $7pm-9pm$, ..., $t=7$ to $5am-7am$ the next day. The sequence of non-EV loads is a normalized version of the global consumption profile in France for the aforementioned period of time; the data are available on the RTE website "http://clients.rte-france.com/lang/fr/visiteurs/vie/courbes.jsp". The other parameters are set to $C=3$ and $P=0.2$ which corresponds to the case where EVs need to charge until $75$\% of their battery capacities\footnote{A full charging of a battery of $24$kWh at $3$kW needs $8$hours.} and a penetration rate of approximately $40$\%\footnote{Given that the maximal household electricity consumption is typically of $6$kW and the EV charging rate at home of $3$kW, the ratio $\frac{0.2/3}{1/6}=0.4$ approximates the EV penetration rate in the district under the assumption that all houses consumes at their full contracted power at peak time.}. Finally, the linear charging cost function $f(L)=L$, for which convergence of the replicator dynamics holds, is considered. We first observe the total load during the considered night, as the sum of the non-EV load, the individual nonatomic charging load and the coalition charging load. It can be observed that while individuals charge mainly during the night when the non-EV load is small, coalition's load is more uniformly distributed.

\begin{figure}[!htbp]
\includegraphics[scale=0.6]{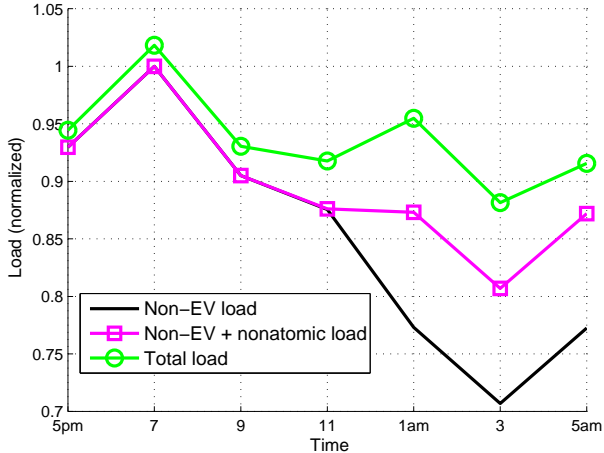}
\caption{\label{fig:totalLoadT7C3} Total load at a CE for EV night charging ($T=7>3$) with a linear charging cost function $f(L)=L$: \textit{individuals' load is put on the "valley" of the non-EV load, coalition's load is more uniformly distributed}.}
\end{figure}

The two following figures are dedicated to study if the main theoretical properties established in the particular case of three time-slots are still observed when considering this larger case.

\begin{figure}[!htbp]
\includegraphics[scale=0.6]{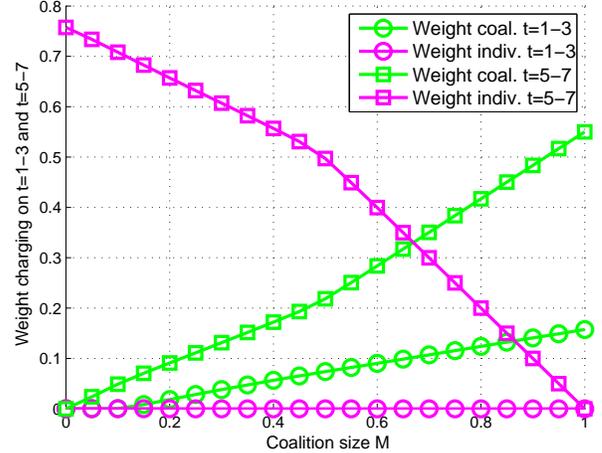}
\caption{\label{fig:configCET7C3} Configuration of CE for EV night charging ($T=7>3$) with a linear charging cost function $f(L)=L$ according to the size of the coalition $M$: \textit{the monotonicity properties theoretically established for $T=3$ and $C=2$ still seem to hold}.}
\end{figure}

Fig. \ref{fig:configCET7C3} (respectively Fig. \ref{fig:costT7C3}) shows that the monotonicity properties of the charging weights (respectively costs) at CE proven in the case $T=3$, $C=2$ still seem to hold: a next step of this work will be to confirm these observations with theoretical arguments. Finally, this again exhibits the phenomenon of "social dilemma".

\begin{figure}[!htbp]
\includegraphics[scale=0.6]{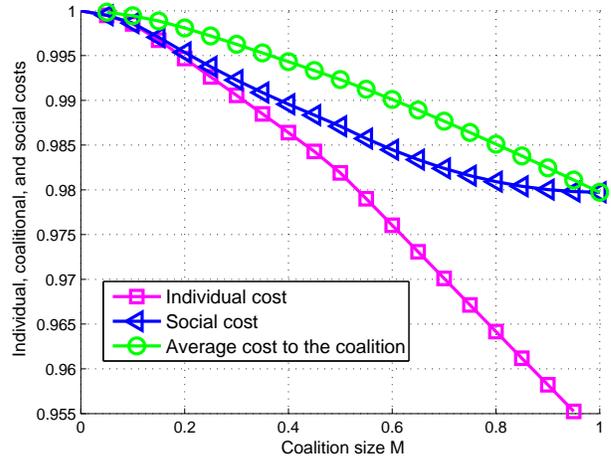}
\caption{\label{fig:costT7C3} Individual, coalitional and social costs at a CE for EV night charging ($T=7>3$) with a linear charging cost function $f(L)=L$ according to the size of the coalition $M$: \textit{the monotonicity properties theoretically established for $T=3$ and $C=2$ still seem to hold}.}
\end{figure}

\section{Conclusion}
\label{sec:ccl}

In this paper, we introduce the game-theoretical framework of composite games as a tool to analyze the situation where both autonomous EVs and coalitions of EVs, i.e., groups of EVs which are coordinated by a unique aggregator, coexist when taking their charging decisions.
In this context, the existence of a stable configuration, a composite equilibrium, is proven to exist. At equilibrium, the cost of individuals, the average cost in the coalition, and the social cost have been compared.

Then, more detailed properties of the composite equilibrium are established in the illustrative case of three time-slots as a first step to validate some intuitions. In particular, it is shown that the charging weights on the time-slot with the largest non-EV demand/load increases with the coalition size for the coalition, while it decreases for the individuals. This highlights the different behaviour of the individuals and the coalition, expressing in particular that larger coalitions integrate more externalities. Furthermore, all the costs are proven to decrease with the size of the coalition supporting the idea of forming big coalitions for EV charging but leading also to a standard ``social dilemma'': The social optimum is obtained for a coalition of maximal size but then each EV prefers to act individually. A relevant extension of this work would be to design incentives to make the configuration with a coalition of maximal size stable.

Finally, simulations both quantify these phenomena in the simple case of three time-slots and are also conducted in the realistic case of EV night charging with a larger number of time-slots. Interestingly, the theoretical results proven in the case of three time slots seem to hold in the simulation realized in this latter case: this shows that there is still room for improving the understanding of the properties of this problem in a general setting. The framework of composite games seems to be particularly promising for understanding heterogenous distributed networks such as smart grids.

\ifCLASSOPTIONcaptionsoff
  \newpage
\fi

\bibliographystyle{IEEEtran_NoURL}
\bibliography{BiblioCompositeSG}

\begin{thebibliography}{10}
\providecommand{\url}[1]{#1}
\csname url@samestyle\endcsname
\providecommand{\newblock}{\relax}
\providecommand{\bibinfo}[2]{#2}
\providecommand{\BIBentrySTDinterwordspacing}{\spaceskip=0pt\relax}
\providecommand{\BIBentryALTinterwordstretchfactor}{4}
\providecommand{\BIBentryALTinterwordspacing}{\spaceskip=\fontdimen2\font plus
\BIBentryALTinterwordstretchfactor\fontdimen3\font minus
  \fontdimen4\font\relax}
\providecommand{\BIBforeignlanguage}[2]{{%
\expandafter\ifx\csname l@#1\endcsname\relax
\typeout{** WARNING: IEEEtran.bst: No hyphenation pattern has been}%
\typeout{** loaded for the language `#1'. Using the pattern for}%
\typeout{** the default language instead.}%
\else
\language=\csname l@#1\endcsname
\fi
#2}}
\providecommand{\BIBdecl}{\relax}
\BIBdecl

\bibitem{Clement2009}
\BIBentryALTinterwordspacing
K.~Clement, E.~Haesen, and J.~Driesen, ``{Coordinated charging of multiple
  plug-in hybrid electric vehicles in residential distribution grids},''
  \emph{2009 IEEE PES Power Systems Conference and Exposition}, vol.~25, no.~1,
  pp. 1--7, 2009.
\BIBentrySTDinterwordspacing

\bibitem{Gong2011}
\BIBentryALTinterwordspacing
Q.~Gong, S.~Midlam-Mohler, V.~Marano, and G.~Rizzoni, ``{Study of PEV Charging
  on Residential Distribution Transformer Life},'' \emph{IEEE Transactions on
  Smart Grid}, vol.~3, no.~1, pp. 404--412, 2011.
\BIBentrySTDinterwordspacing

\bibitem{Galus2008}
M.~D. Galus and G.~Andersson, ``Demand management of grid connected plug-in
  hybrid electric vehicles (phev),'' in \emph{Energy 2030 Conference, 2008.
  ENERGY 2008. IEEE}, 2008, pp. 1--8.

\bibitem{Deilami2011}
\BIBentryALTinterwordspacing
S.~Deilami, A.~S. Masoum, P.~S. Moses, and M.~A.~S. Masoum, ``{Real-Time
  Coordination of Plug-In Electric Vehicle Charging in Smart Grids to Minimize
  Power Losses and Improve Voltage Profile.}'' \emph{IEEE Trans. Smart Grid},
  vol.~2, no.~3, pp. 456--467, 2011.
\BIBentrySTDinterwordspacing

\bibitem{Shinwari12}
M.~Shinwari, A.~Youssef, and W.~Hamouda, ``{A Water-Filling Based Scheduling
  Algorithm for the Smart Grid},'' \emph{Smart Grid, IEEE Transactions on},
  vol.~3, no.~2, pp. 710--719, 2012.

\bibitem{Wu2012}
C.~Wu, A.-H. Mohsenian-Rad, and J.~Huang, ``{Wind Power Integration via
  Aggregator-Consumer Coordination: A Game Theoretic Approach},'' \emph{IEEE
  PES Innovative Smart Grid Technologies Conference}, 2012.

\bibitem{Chao2010}
\BIBentryALTinterwordspacing
H.-p. Chao, ``{Price-Responsive Demand Management for a Smart Grid World},''
  \emph{The Electricity Journal}, vol.~23, no.~1, pp. 7--20, 2010.
\BIBentrySTDinterwordspacing

\bibitem{Saad2012}
W.~Saad, Z.~Han, H.~V. Poor, and T.~Basar, ``{Game Theoretic Methods for the
  Smart Grid},'' \emph{CoRR}, vol. abs/1202.0, 2012.

\bibitem{Agarwal2011}
T.~Agarwal and S.~Cui, ``{Noncooperative Games for Autonomous Consumer Load
  Balancing over Smart Grid},'' \emph{CoRR}, vol. abs/1104.3, 2011.

\bibitem{Mohsenian-Rad2010}
A.-H. Mohsenian-Rad, V.~W.~S. Wong, J.~Jatskevich, R.~Schober, and
  A.~Leon-Garcia, ``{Autonomous Demand Side Management Based on Game-Theoretic
  Energy Consumption Scheduling for the Future Smart Grid},'' \emph{Smart Grid,
  IEEE Transactions}, vol.~1, no.~3, pp. 320--331, 2010.

\bibitem{Ibars2010}
C.~Ibars, M.~Navarro, and L.~Giupponi, ``Distributed demand management in smart
  grid with a congestion game,'' in \emph{Smart Grid Communications
  (SmartGridComm), 2010 First IEEE International Conference on}, 2010, pp.
  495--500.

\bibitem{Han2010}
S.~Han, S.~Han, and K.~Sezaki, ``{Development of an Optimal Vehicle-to-Grid
  Aggregator for Frequency Regulation},'' \emph{IEEE Trans. Smart Grid},
  vol.~1, no.~1, pp. 65--72, 2010.

\bibitem{Wan11}
C.~Wan, ``Coalitions in nonatomic network congestion games,'' \emph{Math. Oper.
  Res.}, vol.~37, no.~4, pp. 654--669, 2012.

\bibitem{Beaude2012}
O.~Beaude, S.~Lasaulce, and M.~Hennebel, ``{Charging games in networks of
  electrical vehicles},'' in \emph{Network Games, Control and Optimization
  (NetGCooP), 2012 6th International Conference on}, 2012, pp. 96--103.

\bibitem{Orda1993}
A.~Orda, R.~Rom, and N.~Shimkin, ``{Competitive Routing in Multi-User
  Communication Networks},'' \emph{IEEE/ACM Transactions on Networking},
  vol.~1, pp. 510--521, 1993.

\bibitem{Gan2012}
L.~Gan, U.~Topcu, and S.~H. Low, ``{Stochastic distributed protocol for
  electric vehicle charging with discrete charging rate},'' in \emph{2012 IEEE
  Power and Energy Society General Meeting}, vol.~27, no.~3, IEEE.\hskip 1em
  plus 0.5em minus 0.4em\relax IEEE, 2012, pp. 1--8.

\bibitem{War52}
J.~G. Wardrop, ``Some theoretical aspects of road traffic research,'' in
  \emph{{Proceedings of the Institute of Civil Engineers, Part II}}, vol.~1,
  1952, pp. 325--378.

\bibitem{Kin86}
D.~Kinderlehrer and G.~Stampacchia, \emph{{An Introduction to Variational
  Inequalities and Their Applications}}.\hskip 1em plus 0.5em minus 0.4em\relax
  SIAM, 2000.

\bibitem{Mertikopoulos09}
\BIBentryALTinterwordspacing
P.~Mertikopoulos and A.~L. Moustakas, ``Rational behaviour in the presence of
  stochastic perturbations,'' \emph{CoRR}, vol. abs/0906.2094, 2009.
\BIBentrySTDinterwordspacing

\bibitem{Cominetti09}
\BIBentryALTinterwordspacing
R.~Cominetti, J.~R. Correa, and N.~E.~S. Moses, ``The impact of oligopolistic
  competition in networks.'' \emph{Operations Research}, vol.~57, no.~6, pp.
  1421--1437, 2009.
\BIBentrySTDinterwordspacing

\end{thebibliography}

\end{document}